\pdfoutput=1
\documentclass[11pt]{article}

\usepackage{ltexpprt}
\usepackage[margin=1in]{geometry}

\usepackage{imports}
\usepackage{csquotes}

\creflabelformat{equation}{#2\textup{#1}#3}

\usepackage{marginnote}
\newif\ifnotes
\def\noteson{%
  \notestrue
    \gdef\SRI##1{
        \marginnote{\color{red} \footnotesize{##1}}
    }
}

\noteson

\ifnotes

\else

\fi

\newcommand{\diag}{\textrm{diag}}

\newcommand{\optval}{\ensuremath{\alocal}}
\newcommand{\alocal}{\ensuremath{\alpha_{\textsc{local}}}}
\newcommand{\compratio}{\ensuremath{O(\optval)}}
\newcommand{\tO}{\ensuremath{\widetilde{O}}}
\newcommand{\weightinv}{\ensuremath{\left( p_e + \nicefrac{1}{m} \right)}}
\newcommand{\weight}{\ensuremath{\weightinv^{-1}}}

\newcommand{\stre}{\normalfont\textrm{stretch}}
\newcommand{\strew}{\stre_w}

\newcommand{\load}{\normalfont\textrm{load}}
\newcommand{\loadw}{\load_w}
\newcommand{\apxload}{\normalfont\textrm{aload}}
\newcommand{\apxloadw}{\apxload_w}
\newcommand{\apx}[1]{\textsc{approx}(#1)}

\newcommand{\sketchmatrix}{\textsc{SketchMatrix}}
\newcommand{\lapsolve}{\textsc{ApproxLapSolve}}
\newcommand{\xlapsolve}{\textsc{LapSolve}}
\newcommand{\recovernorm}{\textsc{RecoverNorm}}

\newcommand{\pe}[1]{\ensuremath{p_e^{(#1)}}}
\newcommand{\xe}[1]{\ensuremath{x_e^{(#1)}}}
\newcommand{\we}[1]{\ensuremath{w_e^{(#1)}}}

\newcommand{\xv}[1]{\ensuremath{x^{(#1)}}}
\newcommand{\XX}[1]{\ensuremath{X^{(#1)}}}
\newcommand{\WW}[1]{\ensuremath{W^{(#1)}}}
\newcommand{\MM}[1]{\ensuremath{M^{(#1)}}}
\newcommand{\potential}[1]{\ensuremath{\| \xv{#1} \|_1}}

\newcommand{\eLf}{\ensuremath{b_e L\pinv b_f\tp}}
\newcommand{\aeLf}{\ensuremath{|\eLf|}}

\newcommand{\tU}{\ensuremath{\widetilde{U}}}

\newcommand{\sketchd}{\ensuremath{\delta}}
\newcommand{\sketche}{\ensuremath{\epsilon}}
\newcommand{\lsolved}{\ensuremath{\delta}}
\newcommand{\lsolvee}{\ensuremath{\epsilon_L}}
\newcommand{\loade}{\ensuremath{\epsilon}}
\newcommand{\ComputeRouting}{\textsc{ComputeRouting}}
\newcommand{\GetApproxLoad}{\textsc{GetApproxLoad}}

\newcommand{\assumeexact}{assuming exact Laplacian solver}

\newcommand{\median}{\textsc{median}}
\newcommand{\edash}{\ensuremath{\epsilon'}}
\newcommand{\medianx}{\ensuremath{\median(|x_1|, |x_2|, \ldots, |x_\ell|)}}
\newcommand{\mediany}{\ensuremath{\median(|y_1|, |y_2|, \ldots, |y_\ell|)}}

\newcommand{\alt}{\ensuremath{\apxload_{w^{(t)}}}}
\newcommand{\lt}{\ensuremath{\load_{w^{(t)}}}}

\newcommand{\ope}{\ensuremath{(1+\epsilon)}}
\newcommand{\ome}{\ensuremath{(1-\epsilon)}}
\newcommand{\opet}{\ensuremath{\left(1+\frac{\epsilon}{2}\right)}}
\newcommand{\omet}{\ensuremath{\left(1-\frac{\epsilon}{2}\right)}}
\newcommand{\lb}{\ensuremath{n^{-2}}}
\newcommand{\lbnm}{\ensuremath{n^{2}}}
\newcommand{\ub}{\ensuremath{2n}}
\newcommand{\totalb}{\ensuremath{2n^3}}
\newcommand{\dem}{d}

\newcommand{\tR}{\ensuremath{\widetilde{R}}}
\newcommand{\tC}{\ensuremath{\widetilde{C}}}
\newcommand{\largest}{\ensuremath{K}}

\usepackage[colorinlistoftodos,prependcaption,textsize=tiny]{todonotes}

\usepackage{authblk}

\begin{document}
\title{Electrical Flows for Polylogarithmic Competitive Oblivious Routing}
\author[1]{Gramoz Goranci}
\affil[1,5]{Faculty of Computer Science, University of Vienna, Austria}
\author[2]{Monika Henzinger}
\affil[2]{Institute of Science and Technology Austria (ISTA), Klosterneuburg, Austria}
\author[3]{Harald R\"acke}
\affil[3]{Technical University Munich, Germany}
\author[4]{Sushant Sachdeva}
\affil[4]{University of Toronto, Canada}
\author[5]{A. R. Sricharan}
\affil[5]{UniVie Doctoral School Computer Science DoCS}

\date{}

\maketitle

\SetTracking{encoding=*, shape=sc}{20}

\begin{abstract}
Oblivious routing is a well-studied paradigm that uses static precomputed routing tables for selecting routing paths within a network. Existing oblivious routing schemes with polylogarithmic competitive ratio for general networks are tree-based, in the sense that routing is performed according to a convex combination of trees. However, this restriction to trees leads to a construction that has time quadratic in the size of the network and does not parallelize well.

In this paper we study oblivious routing schemes based on electrical routing. In particular, we show that general networks with $n$ vertices and $m$ edges admit a routing scheme that has competitive ratio $O(\log^2 n)$ and consists of a convex combination of only $O(\sqrt{m})$ electrical routings. This immediately leads to an improved construction algorithm with time $\tilde{O}(m^{3/2})$ that can also be implemented in parallel with $\tilde{O}(\sqrt{m})$ depth. \end{abstract}

\section{Introduction}

Oblivious routing schemes use static-precomputed routing tables for selecting
routing paths instead of routing paths that adapt dynamically  to the observed
traffic pattern of a parallel system. While at first glance this restriction
seems like a serious barrier to obtaining good performance, it has been shown
that in undirected networks oblivious routing does not only provide good
theoretical guarantees \cite{Rae02,HHR03,Rae08},
but is also an excellent choice
for practical implementations \cite{AC03,KYY+18}
due to its simple structure.

Formally an oblivious routing scheme provides a unit flow $f_{s,t}$ between
every source-target pair $(s,t)$ in the network. When a demand $d_{s,t}$
between $s$ and $t$ appears, this unit flow is scaled by the demand to provide
the required flow between source and target. When using an oblivious routing scheme for packet routing the path of a
packet is chosen according to the flow so that the probability that the packet
takes a certain edge is equal to the flow value along that edge. The main
strength of any oblivious routing algorithm stems from the fact that determining
the next hop for a packet can be done via a simple table lookup after
precomputing the necessary routing tables.

In this paper we consider oblivious routing algorithms that aim to minimize
network \emph{congestion}, i.e., the maximum flow on any edge of the network.
Most existing oblivious routing algorithms for this congestion cost-measure are \emph{tree based}
\cite{Rae02,Rae08,HHR03,RST14,Englert2009}. This means a convex combination of trees is
embedded into the network. A routing path between two vertices $s$ and $t$ is
then in principle chosen by first sampling a random tree and then following the
path between $s$ and $t$ in this tree.

Another approach for oblivious routing is to use \emph{electrical flows}. In an
electrical flow routing one assigns a \emph{resistance} (or its inverse, which is called \emph{conductance}) to each edge of
the graph. The flow between two vertices $s$ and $t$ is then defined by the
current that would result when adding a voltage source between $s$ and $t$.
Lawler and Narayanan~\cite{LN09mixing} studied the performance of electrical flows as an oblivious routing scheme. They show that if every edge is assigned unit resistance, then electrical routing has a \emph{competitive ratio}
of $\min\{\sqrt{m},O(T_{\mathrm{mix}})\}$
against any $\ell_p$-norm \emph{simultaneously}, where the competitive ratio is the ratio of the cost of routing any demand using oblivious routing to the cost of the optimal routing for that demand, and
$T_{\mathrm{mix}}$ is the mixing time of a random walk on the graph. %
Kelner and Maymounkov~\cite{KM11electric} consider electrical
routing on expander graphs (again with uniform resistances) and show that this scheme achieves small competitive ratio and is quite robust under edge deletions. Schild, Rao, and Srivastava~\cite{Schild2018} study the average length of electrical flow paths, and as a consequence obtain $O(\log^2 n)$ competitive ratio for the special family of \emph{edge-transitive} graphs. However, no non-trivial bound on the competitive ratio is known for electrical routing on general networks. This leads to the following important question:

\begin{center}
\emph{Does there exist an electrical routing (with \emph{non-uniform} conductances) that obtains a polylogarithmic competitive ratio with respect to the congestion cost-measure?}
\end{center}

In this paper we give a partial answer to this question by showing how to obtain a
competitive ratio of $O(\log^2n)$ w.r.t.\ congestion with a \emph{convex combination}
of only $O(\sqrt{m})$ electrical flows\footnote{In personal communication, Sidford and Lee~\cite{SL22} claimed that they had also observed that there exists an oblivious routing scheme based on a convex combination of $O(\sqrt{m})$ electrical flows that achieves $O(\log^2 n)$ competitive ratio.}.
\begin{restatable}{theorem}{MainThm}
\label{thm:MainThm}
Given a capacitated graph $G=(V,E)$ with $n$ vertices and $m$ edges, there is an algorithm that finds an oblivious routing scheme composed of a convex combination of $O(\sqrt{m})$ electrical flows. The algorithm has competitive ratio $O(\log^2n)$, runs in time $\tO(m^{3/2})$, and can be implemented in parallel with $\tO(\sqrt{m})$ depth.
\end{restatable}

More specifically, we show that our competitive ratio is proportional to the best bound on the \emph{localization of electrical flows} (\cref{lem:Localization}), which is currently shown to be $O(\log^2 n)$. Hence, any improvement on the localization bound would lead to an improvement in our competitive ratio.

 In contrast to this, all existing tree based oblivious routing schemes
that guarantee a polylogarithmic competitive ratio require at least $\Omega(m)$
trees\footnote{The routing schemes in~\cite{Rae02} and~\cite{RST14} are based
on a single tree with different embeddings into the network; for this
discussion, this is viewed as several trees.}.
The best achievable competitive ratio with tree based schemes is $O(\log n)$.
Note that electrical flows in particular generalize oblivious routings that are based
on a convex combination of \emph{spanning trees}\footnote{In general tree based
routings may embed trees that contain Steiner nodes, i.e., vertices that are
not part of the network, and they may also use arbitrary paths to connect
embedded vertices.}
(as one could simply give all edges that are not part of
the spanning tree a resistance of infinity). Therefore, our work shows that going from trees
to electrical flows allows us to reduce the support of the convex combination
from $O(m)$ to $O(\sqrt{m})$ with a slight increase in competitive ratio.

This reduction in the support of the convex combination also has an important
implication for the construction time, and in particular for the parallel depth
of the construction.
The state of the art for constructing tree based oblivious routing schemes
(with a polylogarithmic guarantee on the competitive ratio) is the
multiplicative weights update method~\cite{AHK12}.
One iteration computes
a low stretch spanning tree in time $\tO(m)$ and updates weights/distances
on edges. This results in depth $\tO(m)$ and work $\tO(m^2)$ for computing
$O(m)$ trees. R\"acke, Shah and T\"aubig~\cite{RST14} obtained a near-linear time algorithm for computing a single tree flow sparsifiers with quality $O(\log^{4} n)$. While their construction can be adapted to an $O(\log^{4} n)$ competitive oblivious routing scheme, the best-known efficient implementations of such a scheme require at least quadratic time in the size of the network.

We use the multiplicative weights update method with electrical flows. One
iteration computes a set of resistances/conductances that is good on average
for the current edge-weights and updates weights for the next iteration. We
show that by using matrix sketching we can implement one iteration in
time $\tO(m)$ as in the tree case. Because electrical flows minimizes the $\ell_2$-squared norm of the flow values, we can show that we only require $O(\sqrt{m})$ iterations to obtain a good oblivious routing scheme.
This results in depth $\tO(\sqrt{m})$
and work $\tO(m^{3/2})$ for computing the convex combination of electrical
flows\footnote{Note that computing the routing tables for quickly determining
the outgoing edge from the header of a packet adds additional work
for tree-based and electrical-flow-based routing. For tree based routing this
can be done with work $\tO(n^2)$ per tree and depth $\tO(1)$. For electrical
flow based routing the work is $\tO(mn)$ per flow with depth $\tO(1)$.
See Section~\ref{sec:repPRAM}.}.

We show that our techniques also apply to oblivious routing measured with respect to the $\ell_1$-norm of loads on the edges,
where we
obtain a competitive ratio of $O(\log^2n)$ with $O(\sqrt{m})$ electrical flows.
See \cref{sec:l1} for details.

\subsection*{Related work}
The study of oblivious routing was initiated by Valiant and Brebner~\cite{VB81}, who
developed an oblivious routing protocol for routing in the hypercube that
routes any permutation in time that is only a logarithmic factor away from
optimal. For the cost measure of minimizing congestion,
R\"acke~\cite{Rae02} proved the existence of an oblivious routing scheme with
polylogarithmic competitive ratio for any undirected network. This result was
subsequently improved to a competitive ratio of
$O(\log^2n \log\log n)$~\cite{HHR03}, and then to a competitive ratio of
$O(\log n)$~\cite{Rae08}, which is optimal.
Englert and Räcke~\cite{Englert2009} extend these results to oblivious routing
when the cost-measure is the $\ell_p$-norm of the edge congestions.

Oblivious routing by electrical flows was first considered by Harsha et
al.~\cite{HHN+08} for the goal of designing oblivious routing schemes that minimize
the cost-measure $\|\cdot\|_2^2$, also known as \emph{average latency}, for the case of a single target.
In this scenario, using electrical flows is a very intuitive approach
as it minimizes the $\|\cdot\|_2^2$ of the flow. Specifically, in their algorithm, every
source in the oblivious routing scheme optimizes its flow as if no other source
was active. They show that this gives a competitive ratio of $O(\log n)$ for the $\ell_2$-norm squared of the flow.

Electrical flows also play a major role as a tool for speeding up flow
computations (see e.g.\ \cite{Mad16,BGJ+22,GaoLP21,DGG22,AMV20,LRS13,KLY+14,CKM+11}). This is
due to the fact that electrical flow computations reduce to solving Laplacian
systems, which is a task that can be performed in nearly linear time. In addition, flow algorithms
usually aim to minimize the $\ell_{\infty}$-norm in some way. The fact that the
$\ell_2$-norm is closer to the $\ell_{\infty}$ than e.g.\ the $\ell_1$ norm
helps in these optimizations.
Electrical flows have also been used to generate alternative routes in road networks~\cite{SFG21}.

 Ghaffari, Haeupler, and Zuzic~\cite{GHZ21}
 introduced the concept of \emph{hop-constrained} oblivious routing. This is an
 oblivious routing scheme that is given an additional parameter $h$ that can be
 viewed as an upper bound on the dilation used by an optimum routing. They show
 that one can obtain an oblivious routing scheme that guarantees a congestion
 of $\tO(C_h)$ and a dilation of $\tO(h)$, where $C_h$ is the optimum
 congestion that can be obtained with routing paths of length at most $h$.
 However, their construction still suffers from the same drawback as tree-based
 routing schemes in that it requires $\Omega(m)$ iterations of multiplicative
 weights, which results in a depth of $\Omega(m)$. There has also been recent interest in computing competitive routing schemes with small support, as shown in the work of Zuzic, Haeupler, and Roeyskoe~\cite{zuzic2023sparse} in the context of \emph{semi-oblivious} routing.

 In the sub-polynomial competitive ratio regime, Kelner et al.~\cite{KLY+14} give an algorithm to compute an $n^{o(1)}$ competitive oblivious routing in $m^{1+o(1)}$ time, and use this as a subroutine to compute a $(1+\epsilon)$-approximate maximum $s$-$t$ flow and concurrent multicommodity flow in almost-linear time.
 Haeupler et al.~\cite{HRG22} present a construction of routing tables for
 $h$-hop routing that runs in $O(D+\mathrm{poly}(h))$ rounds of the
 CONGEST-model of distributed computing, and guarantee a competitive ratio of
 $n^{o(1)}$. This essentially means that the resulting packet routing algorithm
 can schedule any permutation in time $n^{o(1)}$ whenever this is possible, and it's extremely fast.
 However, the scheme crucially uses the fact that routing requests are
 initiated from both communication partners; it requires name-dependent routing
 (the node ids are changed during initialization to allow for compact routing
 tables), and it does not achieve a polylogarithmic competitive ratio.

\subsection*{Technical Overview}
For the so-called \emph{linear oblivious routing schemes} (such as tree-based routing schemes
or electrical flows) it is well known that the worst demand is a demand of $1$
along every edge in the unweighted network. An optimal algorithm can trivially route this demand by
sending $1$ unit of flow along every edge in the network, and, hence, has maximum congestion $1$. Let
$\load_{w}(e)$ denote the load on edge $e$ generated by an electrical flow
routing (with conductances $w$) when routing this worst case demand. The
competitive ratio of the routing scheme is then $\max_e\load_{w}(e)$. We can
write the search for a good convex combination of electrical flow routings $w_i$ as
a linear program:
\begin{equation*}\label{loadlp}%
\begin{array}{@{}ll@{}rl}
 \text{min}  &  \alpha &  &   \\[0.1cm]
 \text{s.t.}  & \forall e \quad  \quad &   \sum_{i} \nolimits \lambda_i \cdot \load_{w_i}(e) &\leq \alpha \\[0.2cm]
  & &  \sum_{i} \nolimits \lambda_i & = 1  \\[0.2cm]
  & \forall i \quad \quad &  \lambda_i & \geq 0\enspace.
\end{array}
\end{equation*}
We want to find a good solution to this problem, say  with
objective value $\beta$.
For this the multiplicative weights update method maintains a set of
weights $p_e \geq 0$ on the edges with $\sum_ep_e=1$. In each iteration
it computes \emph{one} routing $R_i$ such that $R_i$ fulfills
$\sum_ep_e\load_{R_i}(e)\le\beta$, i.e., the weighted total load over all edges for $R_i$ is at most $\beta$.
Here $\load_{R_i}(e)$ denotes the load induced on edge $e$
when routing a demand of 1 along every edge using the routing $R_i$.
This means \emph{on average} the edges have load at most $\beta$ when using
routing $R_i$. For the next iteration the weight $p_e$ of edges with
$\load_{R_i}(e)>\beta$ is increased while edges with $\load_{R_i}(e)<\beta$
decrease their weight. In the end the convex combination of all the routings
will have load at most $\beta$ for every edge.

To apply this scheme we first need a bound $\beta$ such that
we always can (efficiently) find a routing with
$\sum_ep_e\load_{R_i}(e)\le\beta$. For tree-based routing one can use
small stretch trees~\cite{FRT03,ABN08} to solve this problem with $\beta=O(\log n)$.
For electrical routing we show that we can bound $\beta$ using a specific choice of parameters in the so-called \emph{localization
lemma} due to Schild et al.~\cite{Schild2018}.
Informally, the localization lemma states that the average length of flow paths
in an electrical routing is small. In particular, we show that we can find parameters for localization which guarantee the existence of
conductances $w$ such that the load when routing via the
electrical flow with conductances $w$ (say $\load_{w}$) fulfills
$\sum_ep_e\load_{w}(e)\le\beta$, with $\beta=O(\log^2n)$.

Secondly we need to be able to efficiently compute the load on every edge
after computing the routing $R_i$ in an iteration, in order to be able to
adjust the weights for the next iteration. Na\"ively routing across each edge using $R_i$ would require time $O(m^2)$ for solving $m$
Laplacian systems to compute the electrical flow between the end points of each edge. We use a sketching result due to Indyk~\cite{Ind06sketch} for approximating
the $\ell_1$-norm of vectors. This allows us to approximate the loads due to all $m$ electrical flows in time
$\tO(m)$. We show that this approximation is still sufficient for the
multiplicative weights method to work correctly.

The advantage of using electrical flows within the multiplicative weights
update framework is that it allows to derive a better bound on the convergence time.
The number of iterations required for the method is
proportional to how far each inequality, namely $\load_w(e) \le \beta$ for
each $e$, is violated by the routing found in each iteration. This is the
so-called \emph{width} of the algorithm, and we bound the width by
$O(\sqrt{m})$ by regularizing the weights to be nearly uniform.
This in turn implies that we
require only $O(\sqrt{m})$ iterations.

\section{Preliminaries}
\label{sec:Preliminaries}

For simplicity of presentation, we present the uncapacitated case in the main body of the paper, and explain the changes required for the capacitated case in \cref{sec:capacitated}.

\subsection{Oblivious Routing}

\paragraph*{Graph}
We will route on
unweighted, undirected graphs. Choose an arbitrary orientation of the edges and let $B$ be the $m \times n$ incidence matrix of $G$. Denote by $b_e$ the row of $B$ corresponding to edge $e$.

\paragraph*{Demand}
A demand $\chi \in \mathbb{R}^n$ encodes the flow deficit/surplus requirements for each vertex in the graph. A demand is \emph{valid} if $\sum_i \chi_i = 0$, that is, the requirements at all the vertices of the graph cancel out. Let $X \subseteq \mathbb{R}^n$ denote the space of all valid demands.
Note that the vector $b_e\tp$ has unit demand across the endpoints of $e$ (with the direction of demand decided by the orientation of the edge). If the demands are only between pairs of vertices $(s,t)$, then we encode the demands between all pairs of vertices in a \emph{demand vector} $ \dem \in \R^{\binom{n}{2}}$.

\paragraph*{Flow}
A flow corresponding to a demand $\chi$ is a vector $f \in \mathbb{R}^m$ that ensures that the flow out of a vertex satisfies the demand, namely, $\chi_v + \sum_{u\sim v} f_{uv} = 0$ for all $v \in V$. For example, an $(s,t)$-flow is one where the demands of every vertex except $s$ and $t$ are zero.

\paragraph*{Oblivious routing}
An $m \times n$ matrix $M$ that maps vertex demands to edge flows is called an \emph{oblivious routing scheme}\footnote{While an oblivious routing scheme in general could be non-linear, we only consider linear routing schemes in this paper.} if  for any valid demand $\chi$  the vertex demands created by the flow generated by $M$ are equal to the input vertex demands. Formally, we require that $B\tp M$ is identity on $X$, i.e., $B\tp M\chi = \chi$ for all valid demands $\chi \in X$.

\paragraph*{Competitive ratio}
For an oblivious routing $M$ and any collection of demands $D= \{\chi_i\}_i$,
the flow across edge $e$ for routing demand $\chi_i$ is given by $|b_e M \chi_i|$. Thus the total flow (without cancellations) across edge $e$ when routing all the demands in $D$ \emph{independently} using $M$ is given by $f_e = \sum_{\chi_i \in D} |b_e M \chi_i|$.
We would like to minimize the $\ell_p$ norm (in this paper, $p \in \{1, \infty\}$) of the flow resulting from routing these demands independently with an oblivious routing $M$, as compared with the optimal routing that can adaptively choose the routing based on the demands.
If $OPT_p(D)$ is the cost resulting from the optimal routing, and $f$ is the flow described above, then the competitive ratio of the oblivious routing scheme $M$ is given by
\begin{equation*}
\beta_p(M) = \max\limits_{D} \frac{ \|f\|_p}{OPT_p(D)}
\end{equation*}
Kelner and Maymounkov~\cite[Theorem 3.1]{KM11electric} show that for a linear oblivious routing on an uncapacitated graph, the worst case demand set is routing one unit of flow across every edge of the graph, i.e., the worst case requests are the columns of $B\tp$. Since the optimal routing is to route each demand across the same edge, $OPT(B\tp) = 1$. For $p \in \{ 1, \infty \}$, the competitive ratio is then given by $ \beta_p(M) = \| MB\tp \|_p $.

\paragraph*{Load}
For oblivious routing $M$, the flow across edge $e$ for the demand $b_f$ is given by $\load_M(f\to e)$. This can be obtained from the matrix $M$ as $\load_M(f\to e) = |(Mb_f\tp)_e|$. The sum of all flows across an edge $e$ for every possible edge demand vector $b_f$ is then the load across the edge $e$.
\[
\load_M(e) = \sum_f \load_M(f \to e) = \sum_f |(Mb_f\tp)_e|
.\]

Note that the congestion, or the competitive ratio for the $\ell_{\infty}$ norm, is then just the maximum load on any edge in the graph by definition of $\| MB\tp \|_{\infty}$.

\paragraph*{Convexity of load}
Load is convex on oblivious routings. For two oblivious routings $M_0$ and $M_1$, let $M_{\lambda} = \lambda M_1 + (1 - \lambda) M_0$ for $\lambda \in (0, 1)$. Then
\begin{align*}
\load_{M_{\lambda}}(f \to e)
&= |((\lambda M_1 + (1- \lambda) M_0)b_f\tp)_e|
\le \lambda |(M_1 b_f\tp)_e| + (1-\lambda)|(M_0b_f\tp)_e| \\
&= \lambda \load_{M_1}(f \to e) + (1 - \lambda) \load_{M_0} (f \to e)
\end{align*}

\paragraph*{The simplex} We use $\Delta^m = \{ x \in \R^m: \sum x_i = 1, x_i \ge 0 \}$ to denote the $m$-simplex. Any $p \in \Delta^m$ then represents a probability distribution over the edges.

\subsection{Electrical Flow}

\paragraph*{Graph}
Let $(G, w)$ be a weighted, undirected graph with edge weights $\{w_e\}$\footnote{Think of them as the conductances of the edges in an electrical network.}. Choose an arbitrary orientation of the edges. $B$ is the $m \times n$ incidence matrix of $G$. If $W = \diag(w)$ is the $m \times m$ diagonal matrix of edge weights, then the $n \times n$ \emph{Laplacian matrix} is given by $L = B\tp W B$, and $L\pinv$ is its pseudoinverse.
Notationally, when we use $L\pinv$, we always mean the pseudoinverse of the Laplacian with respect to the current weighted graph and \emph{not} the underlying unit-weighted graph.

\paragraph*{Electrical flow}
Given a weighted graph $(G,w)$ and a demand $\chi$, the electrical flow corresponding to $\chi$ is given by $WBL\pinv \chi$.

\paragraph*{Load}
By the above characterization of electrical flow, the load across an edge $e$ for unit current demand across $f$ is given by $\loadw(f \to e) = w_e | b_e L\pinv b_f\tp|$. The load on an edge is the sum of all loads on the edge for demands $b_f$ for all edges $f$, given by
\(
\loadw(e) = w_e \sum_f |b_e L\pinv b_f\tp|
.\)

\paragraph*{Weights determine electrical flow}
An electrical flow is completely determined by the edge weights. That is, given an edge weighting $w$ of an unweighted graph $G$, there is a unique routing $WBL\pinv$ corresponding to this network. If the unweighted graph $G$ is fixed, we denote the load due to weights $w$ on an edge $e$ as $\loadw(e)$.

\paragraph*{Electrical flows are oblivious}
Let $G$ be an unweighted graph for which we wish to find an oblivious routing. Then for any weighting $w$ of the graph, the corresponding electrical flow $M_w = WBL\pinv$ is an oblivious routing. This is easy to check, since \(
B\tp WBL\pinv = (B\tp WB) (B\tp W B)\pinv
\).
We will refer to $M_w$ as the \emph{electrical routing with respect to the weights $\{w_e\}$}.

\paragraph*{Localization}
We will repeatedly use the following so-called \emph{localization lemma}. A concrete instantiation with $\alocal = c \log^2(n)$ for some constant $c$ was given by Schild et al.~\cite{Schild2018}.
\begin{restatable}[Localization \cite{Schild2018}]{lemma}{Localization}
\label{lem:Localization}
Let $G$ be a graph with weights $\{w_e\}$. Then for any vector $\ell \in \R_{\ge 0}^m$, \[
    \sum_{e, f \in E} \ell_e \ell_f \sqrt{w_e w_f} |b_e L_G\pinv b_f\tp | \le \alocal \cdot \|\ell\|_2^2
\]
\end{restatable}

\subsection{Matrix sketches}
To speed up our algorithms we use a sketching result by
Indyk~\cite[Theorem~3]{Ind06sketch}, adapted from a formulation by Schild~\cite[Theorem~9.13]{Sch18spanning}.

\begin{restatable}[\cite{Ind06sketch}]{theorem}{SketchApprox}
\label{thm:SketchApprox}
Given $m \in \mathbb{Z}_{\ge 1}$, $\sketchd \in (0,1)$, and $\sketche \in (0, 1)$, there is a sketch matrix $C = \sketchmatrix(m, \sketchd, \sketche) \in \R^{\ell \times m}$ and an algorithm $\recovernorm(s)$ for $s \in \R^{\ell}$
such that the following properties hold:
\begin{itemize}
    \item (Approximation) For any $v \in \R^m$, with probability at least $1 - \sketchd$ over the randomness of \sketchmatrix, the value of $r = \recovernorm(Cv)$ is \[
            (1-\sketche) \|v\|_1 \le r \le (1+\sketche) \|v\|_1
    \]
    \item $\ell = c/\sketche^2 \cdot \log(1/\sketchd)$ for some constant $c>1$
    \item (running time) \sketchmatrix\ and \recovernorm\ take time $O(\ell m)$ and $\poly(\ell)$ respectively.
\end{itemize}
\end{restatable}

We will require that our theorems hold with high probability, and thus we will set $\sketchd$ to be $1/\poly(n)$. We will also set the approximation constant $\sketche$ to be $\nicefrac{1}{2}$, which gives $\ell = O(\log n)$.

\section{Algorithm for routing}

\begin{algorithm}[t]
\SetAlgoLined
\DontPrintSemicolon
\caption{\ComputeRouting, to achieve $\compratio$-competitive oblivious routing.}
\label{alg:MWURoute}
\KwInput{An unweighted graph $G$.}
\KwOutput{An oblivious routing scheme on $G$.}
Set $\rho \gets \sqrt{2m}$ and $T \gets \frac{ 8 \rho \ln m}{\alocal}$\;
Initialize $x_e^{(0)} \leftarrow 1$ for all $e \in E$, and $X^{(0)} \leftarrow m$\;
\For{$t = 1, 2, \ldots, T$} {
    Set $\pe{t} \leftarrow \nicefrac{\xe{t-1}}{\XX{t-1}}$ for all $e \in E$\;
    Set $\we{t} \leftarrow \weight$ for all $e \in E$\;
    Set $\WW{t} \leftarrow \diag(w)$\;
    Set $M^{(t)} \leftarrow W^{(t)} B \left( B\tp W^{(t)} B \right) \pinv$\label{line:Mdef}\;
    Set $\alt \leftarrow \textsc{GetApproxLoad}(G, w, 1/2)$\;
    Set $\xe{t} \leftarrow \xe{t-1} \cdot \left[ 1 + \frac{1}{2\rho} \cdot \alt(e) \right] $\;
    Set $\XX{t} \leftarrow \sum_e \xe{t}$\;
}
\Return{$M^* = \frac{1}{T} \cdot \sum_{t=1}^{T} M^{(t)}$}
\end{algorithm}

\begin{algorithm}[t]
\SetAlgoLined
\DontPrintSemicolon
\caption{\GetApproxLoad, to compute approximate loads for electrical routing.}
\label{alg:xLoadApx}
\KwInput{A graph $G$, weights $\{ w_e \}_{e \in E}$ on the edges, approximation factor $\loade$.}
\KwOutput{Approximation $\{ \apxloadw(e) \}_{e \in E}$ to the load on the edges.}
Let $B$ be the edge-vertex incidence matrix of $G$\;
Let $L := B\tp \diag(w) B$ be the Laplacian matrix\;
Set $C \leftarrow \sketchmatrix(m, n^{-10}, \loade)$\;
Set $X \leftarrow B\tp C\tp$\;
Let $X^{(i)}$ be the $i^{th}$ column of $X$ for all $i \in [\ell]$\;
Set $U^{(i)} \leftarrow \xlapsolve(L, X^{(i)})$ for all $i \in [\ell]$ \label{line:xlapsolve} \;
Set $U \leftarrow (U^{(1)}, U^{(2)}, \ldots, U^{(\ell)})$ \algcomment{$U = (CBL\pinv)\tp$}
Set $\apxloadw(e) \leftarrow w_e \cdot \recovernorm(U\tp b_e)$ for all $e \in E$\;
\Return{$\apxloadw$}
\end{algorithm}

In this section, we give an algorithm that returns a routing which achieves competitive ratio $\compratio$ by taking a convex combination of $\sqrt{m}$ electrical flows, where $\optval$ is the best localization result for electrical flows in \cref{lem:Localization}.
We use the \emph{multiplicative weights update} method (MWU) to obtain this guarantee.
The algorithm is presented in Algorithm~\ref{alg:MWURoute} (\ComputeRouting).
We are essentially solving the following linear program with MWU

\begin{equation*}\label{loadprimal}%
\begin{array}{@{}ll@{}rl}
 \text{min}  &  \alpha &  &   \\[0.2cm]
 \text{s.t.}  & \forall e \quad  \quad &   \sum_{i} \nolimits \lambda_i \cdot \load_{w_i}(e) &\leq \alpha \\[0.3cm]
  & &  \sum_{i} \nolimits \lambda_i & = 1  \\[0.3cm]
  & \forall i \quad \quad &  \lambda_i & \geq 0.
\end{array}
\end{equation*}

The primal asks for a convex combination of electrical routings that gives low competitive ratio, where the coefficients of the convex combination are collected in $\lambda$. The dual of the above linear program is equivalent to solving
\begin{equation}
\label{eq:betastar}
\max_{p \in \Delta^m} \min_i \sum_e p_e \load_{w_i}(e),
\end{equation}
Denote the optimal values of the primal and dual by $\alpha^*$ and $\beta^*$ respectively. Since $\alpha^* = \beta^*$ by strong duality, we can obtain an existential bound on $\alpha^*$ by showing that $\beta^* \le \alocal$. This bound on $\beta^*$ is obtained by bounding (\ref{eq:betastar}) using the weighting returned by \cref{lem:OracleReturn} which returns, for any $p \in \Delta^m$, a  weighting $w$ for which the $p_e$ weighted average load is smaller than $2 \alocal$. The proof of the bound in \cref{lem:OracleReturn} uses localization, and we wish to convert this existential result into an algorithmic one, which we do using MWU in Algorithm~\ref{alg:MWURoute}.

Our MWU algorithm can be described as the following primal-dual algorithm running for $T$ iterations: We initially start with the primal vector $\lambda = 0$ and dual vector $\pe{1} = \nicefrac{1}{m}$. We then build a primal solution
incrementally as follows. At iteration $t$, we look at the dual variables $\pe{t}$ of the edges, and find a particular set of weights $\we{t}$ such that routing with respect to these weights gives low average load with respect to the dual variables, i.e., $\sum_e \pe{t} \lt(e)$ is low. As in the existential case, we show a bound of $\compratio$ on this average load using localization.

Recall that the primal vector $\lambda$ gives the coefficients for a convex combination of electrical routings. We now increase the primal coefficient corresponding to the routing that was found at this iteration, $\lambda_{\MM{t}}$, by $\nicefrac{1}{T}$. We essentially \enquote{add} this routing to the final convex combination that we output at the end of the algorithm.
Now we need to update the dual variables for the next iteration. For the routing returned in the next iteration, we want to reduce the load on edges that had really high load in the current iteration. To this end, we compute the loads induced on each edge (using $\GetApproxLoad$), and adjust the dual variables based on how high the loads for the current routing $\MM{t}$ are.

At the end of $T$ iterations,  $\lambda$ is the uniform combination of all $T$ electrical routings computed during each iteration of the algorithm. The analysis of MWU shows that our dual variable updates ensure low load on all the edges for the returned routing.
The number of iterations $T$ is proportional to how far each primal inequality, namely $\lt(e) \le \alocal$ for each $e$, is violated by the routing found in each iteration. This is the \emph{width} of the algorithm, and we bound the width by $O(\sqrt{m})$.

For the running time, note that evaluating loads for all the edges na\" ively would take time $O(m^2)$ in each iteration, since it would involve calculating $\aeLf$ for every pair of edges $(e, f) \in E^2$. However, the load on each edge can also be expressed as the $\ell_1$ norm of the vector $BL\pinv b_e$.
We use a sketching result by Indyk~\cite{Ind06sketch} stated in Theorem~\ref{thm:SketchApprox},
which gives a sketch matrix $C$ that preserves $\ell_1$ norms up to small multiplicative errors.
Similar ideas have been implicitly used by Schild~\cite{Sch18spanning}, and later in other works~\cite{LS18,FGL21} as well.
Our sketching improves the running time from $O(m^2)$ to $\tO(m)$ per  iteration. Since there are $T = O(\sqrt{m})$ iterations, we get a bound of $\tO(m \sqrt{m})$ for constructing our routing.

For simplicity, we first assume that in \GetApproxLoad\ we have an exact Laplacian solver that runs in time $\tO(m)$, and present below the analysis in \cref{sec:lsolvex}. The analysis when using an approximate Laplacian solver is more technical, and we present it in \cref{sec:apxlsolve}.
The complication to overcome when
using an approximate Laplacian solver in line~\ref{line:xlapsolve} of \GetApproxLoad\ is the following: The input to \recovernorm\ might not necessarily be of the form $Cv$ for some $v \in \R^m$, since we use the approximate Laplacian solver on $B\tp C\tp$, and thus would have obtained $\apx{CBL\pinv} b_e$
at the end. Note that if we instead had $C \apx{BL\pinv b_e}$
this would be fine, since we would get that $\recovernorm (CBL\pinv b_e) \approx \|BL\pinv b_e \|_1$ and $\|B L\pinv b_e\|_1 \approx \|\apx{BL\pinv b_e}\|_1$, and the resulting guarantee would be the product of these two approximation guarantees.

At a high level, we have an approximation to the vector \emph{after} sketching ($\apx{Cv}$) instead of an approximation to the vector \emph{before} sketching ($C\apx{v}$). Since $\recovernorm$ is not a norm, we do not have a guarantee that it behaves well on inputs that are close to each other. We need to argue that $\recovernorm$ still returns a useful answer when given as input a vector that does not belong to the column space of $C$, but is nevertheless close in the $L$ norm to the required vector $Cv$. This requires some technical analysis.

\section{Proof of correctness}
\label{sec:lsolvex}

The proof is an adaptation of the multiplicative weight update proof to our setting.
The proof uses the following three lemmas. Due to their length we defer their proofs to Sections~\ref{sec:proofOracleReturn}, \ref{sec:proofWidthBound}, and \ref{sec:proofApxLoad} respectively.
The first lemma shows that the \emph{($p_e$ weighted) average edge load} is $\compratio$.
\begin{restatable}{lemma}{OracleReturn}
\label{lem:OracleReturn}
For any probability distribution $p \in \Delta^m$, the oblivious routing $M_w$ corresponding to the electrical network with weights $w_e = \weight$ satisfies $\sum_e p_e \loadw(e) \le 2\alocal$.
\end{restatable}

We use the second lemma to show that in each iteration of \ComputeRouting\, the true loads on each edge are $O(\sqrt{m})$ away from $\alocal$.

\begin{restatable}{lemma}{WidthBound}
\label{lem:WidthBound}
For any probability distribution $p \in \Delta^m$, the oblivious routing $M_w$ corresponding to the electrical network with weights $w_e = \weight$ satisfies $\loadw(e)  \le \sqrt{2m}$ for every edge $e$.
\end{restatable}

The third lemma shows that the error introduced by using a matrix sketch in $\GetApproxLoad$ is not too large, i.e., that $\GetApproxLoad$ approximates the true loads well. For simplicity, we assume that the Laplacian solver used in $\GetApproxLoad$ is exact and runs in time $\tO(m)$. The analysis with an approximate solver is given in \cref{sec:apxlsolve}.

\begin{restatable}{lemma}{ApxLoad}
\label{lem:ApxLoad}
For any approximation factor $0 < \loade < 1$, and any weighted graph $(G,w)$, let $\loadw = \left( w_e \sum_f \aeLf \right)_{e \in E} $ be the true loads, and $\apxloadw = \GetApproxLoad(G, w, \loade)$ be the approximate loads computed by the algorithm. Then with probability $\ge 1 - \nicefrac{1}{\poly(n)}$,
\[
(1 - \loade) \cdot \loadw(e)  \le \apxloadw(e) \le (1 + \loade) \cdot \loadw(e)
\qquad \text{for all $e \in E$}
\]
\end{restatable}

We first track the potential function $\| \xv{t} \|_1$,  and upper bound the increase in $ \| \xv{t} \|_1$ through the algorithm.

\begin{restatable}{lemma}{xnormub}
\label{lem:xnormub}
For any $t \ge 1$, $\potential{t} \le \potential{t-1} \cdot \exp(2 \alocal /\rho)$. At the end of the algorithm, $\potential{T} \le m \cdot \exp(2\alocal T/\rho)$.
\end{restatable}
\begin{proof}
The implication for $\potential{T}$ follows from the former statement and noting that $\potential{0} = m$.
For any $t\ge 1$, we have
\begin{align*}
\potential{t}
&= \sum_e \xe{t}
= \sum_e \xe{t-1} \cdot \left( 1 + \frac{1}{2\rho} \cdot \alt(e) \right) \\
&= \sum_e \xe{t-1} + \frac{1}{2\rho} \cdot \sum_e \xe{t-1} \cdot \alt(e) \\
&\le \sum_e \xe{t-1} + \frac{\potential{t-1}}{\rho} \cdot \sum_e \frac{\xe{t-1}}{\potential{t-1}} \cdot \lt(e) \tag*{(by \cref{lem:ApxLoad})} \\
&\le \potential{t-1} + \frac{\potential{t-1}}{\rho} \cdot 2 \alocal \tag*{(by \cref{lem:OracleReturn})}\\
&\le \potential{t-1} \cdot \exp\left( \frac{2 \alocal}{\rho} \right) \tag*{\qedhere}
\end{align*}
\end{proof}

Next, we lower bound the weight $\xe{t}$ in terms of the load on each edge.

\begin{restatable}{lemma}{xelb}
\label{lem:xelb}
Let $M^*$ be the routing returned by the algorithm. For any edge $e$,
\[
    \xe{T} \ge \exp\left( \frac{T}{8\rho} \cdot \load_{M^*}(e) \right).
\]
\end{restatable}
\begin{proof}
For any edge $e$ and $t \ge 1$, we have
\begin{align*}
\xe{t}
&= \prod_{t' = 1}^t \left( 1 + \frac{1}{2\rho} \cdot \apxload_{w^{(t')}}(e) \right) \\
&\ge \prod_{t' = 1}^t \left( 1 + \frac{1}{4\rho} \cdot \load_{w^{(t')}}(e) \right) \tag*{(by \cref{lem:ApxLoad})} \\
&\ge \prod_{t' = 1}^t \exp\left( \frac{1}{8\rho} \cdot \load_{w^{(t')}}(e) \right) \tag*{(since $e^x \le 1 + 2x$ for $0 < x < 1$, and \cref{lem:WidthBound})} \\
&= \exp\left( \frac{1}{8\rho} \cdot \sum_{t' = 1}^t \load_{w^{(t')}}(e) \right)
\end{align*}
In particular, for $t = T$, we have
\begin{align*}
\xe{T}
&\ge \exp\left( \frac{1}{8\rho} \cdot \sum_{t' = 1}^T \load_{w^{(t')}}(e) \right) \\
&\ge \exp\left( \frac{T}{8\rho} \cdot \load_{M^*}(e) \right) \tag*{(by convexity of load) \qedhere}
\end{align*}
\end{proof}

\begin{restatable}{theorem}{CompRatio}
\label{thm:CompRatio}
The routing returned by Algorithm~\ref{alg:MWURoute} has competitive ratio $\compratio$.
\end{restatable}
\begin{proof}
Combining Lemmas~\ref{lem:xnormub} and \ref{lem:xelb}, for any edge $e$,
\begin{align*}
m \exp\left( \frac{2 \alocal T}{\rho} \right)
&\ge \potential{T}
\ge \xe{T}
\ge \exp\left( \frac{T}{8 \rho} \cdot \load_{M^*}(e) \right)
\end{align*} which in particular gives the required upper bound on $\load_{M^*}(e)$ for any edge $e$, since
\begin{align*}
\load_{M^*}(e)
&\le 16 \alocal + \frac{8 \rho \log m}{T} \\
&= O(\alocal) \tag*{\qedhere}
\end{align*}
\end{proof}

\subsection{Bound on Average Loads}
\label{sec:proofOracleReturn}

\OracleReturn*
\begin{proof}
Setting $\ell_e$ to $\nicefrac{1}{\sqrt{w_e}}$ and applying~\cref{lem:Localization}, we get
\begin{align*}
\sum_e p_e \loadw(e)
&= \sum_e p_e \sum_{f} \loadw(f\to e) \\
&= \sum_e p_e \cdot w_e \cdot \sum_{f} |b_e L\pinv b_f\tp| \tag*{(by definition of load)}\\
&= \sum_e p_e \cdot \weight \cdot \sum_{f} |b_e L\pinv b_f\tp| \tag*{(by definition of $w_e$)}\\
&\le \sum_e \sum_{f} |b_e L\pinv b_f\tp| \\
&= \sum_{e, f} \nicefrac{1}{\sqrt{w_e w_f}} \cdot \sqrt{w_ew_f} \cdot |b_e L\pinv b_f\tp | \\
&= \sum_{e, f} \ell_e \ell_f \cdot \sqrt{w_e w_f} \cdot |b_e L\pinv b_f\tp | \tag*{(by choice of $\ell_e$)}\\
&\le \alocal \cdot \|\ell\|_2^2 \tag*{(by \Cref{lem:Localization})} \\
&= \alocal \cdot \sum_e \nicefrac{1}{w_e} \\
&= \alocal \cdot \sum_e \left( p_e + \nicefrac{1}{m} \right) \\
&= 2 \alocal,
\end{align*} as required.
\end{proof}

\subsection{Bounding the Width}
\label{sec:proofWidthBound}
Next we show that the width is bounded above by $O(\sqrt{m})$.
We first prove a couple of properties of the $\Pi$ matrix that we will use in our analysis.

\begin{lemma} \label{lem:Pi_is_Projection}
The matrix $\Pi$ is a projection matrix.
\end{lemma}
\begin{proof}
We show that $\Pi^2 = \Pi$.
\begin{align*}
\Pi^2
&= \left( W^{\nicefrac{1}{2}} B L\pinv B\tp W^{\nicefrac{1}{2}} \right) \cdot \left( W^{\nicefrac{1}{2}} B L\pinv B\tp W^{\nicefrac{1}{2}} \right) \\
&= \left( W^{\nicefrac{1}{2}} B \right) \cdot \left( L\pinv B\tp W B L\pinv \right) \cdot \left( B\tp W^{\nicefrac{1}{2}} \right) \\
&= \left( W^{\nicefrac{1}{2}} B \right) \cdot L\pinv \cdot \left( B\tp W^{\nicefrac{1}{2}} \right)%
= \Pi.
\end{align*} as required.
\end{proof}

We will use the next lemma to bound the width for both the $\ell_{\infty}$ and the $\ell_1$ case.

\begin{restatable}{lemma}{widthHelpfulLemma}
\label{lem:widthHelpfulLemma}
Let $G$ be a graph with weights $\{w_e\}$ and let $L$ be the Laplacian matrix associated with $G$. For any edge $e$, we have that
\[
    \sum_f w_e w_f |b_e L\pinv b_f \tp|^2 \leq 1.
\]
\end{restatable}
\begin{proof}
    By Lemma~\ref{lem:Pi_is_Projection} we know that $\Pi$ is a projection matrix. This in particular implies that the diagonal entries of $\Pi$ (and hence $\Pi^2$) are less then $1$. Thus,
\begin{align*}
\sum_f w_e w_f |b_e L\pinv b_f \tp|^2
&= \sum_f \left(\sqrt{w_e w_f} \cdot |b_e L\pinv b_f \tp|\right)^2 \\
&= \sum_f \Pi(e, f)^2
= \Pi^2(e,e)
\le 1,
\end{align*}
which was what we were after.
\end{proof}

\WidthBound*
\begin{proof}
Fixing an edge $e$,
\begin{align*}
\loadw(e)
&= w_e \sum_f |b_e L\pinv b_f\tp|
\le \sum_f \sqrt{\nicefrac{w_e}{w_f}} \cdot \sqrt{w_e w_f} |b_e L\pinv b_f\tp| \\
&\le \sqrt{\sum_f \nicefrac{w_e}{w_f}} \cdot \sqrt{\sum_f w_e w_f |b_e L\pinv b_f\tp|^2} \tag*{(by Cauchy-Schwarz)}
\end{align*}

We now bound each of the two terms separately. For the first term, note that
\begin{align*}
w_e \cdot \sum_f \nicefrac{1}{w_f}
= \weight \cdot \sum_f \left( p_f + \nicefrac{1}{m}\right)
\le 2 \cdot \weight
\le 2 \cdot \nicefrac{1}{\left(\nicefrac{1}{m}\right)} = 2m.
\end{align*}

For the second term, by Lemma~\ref{lem:widthHelpfulLemma}, we know that
\[
    \sum_f w_e w_f |b_e L\pinv b_f \tp|^2 \leq 1.
\]

Putting these two inequalities together, we get that $\loadw(e) \le \sqrt{2m}$ for any edge $e$, which gives the desired bound of $\sqrt{2m}$ on the width.
\end{proof}

\subsection{Proof of \GetApproxLoad\ correctness \assumeexact}
\label{sec:proofApxLoad}

We use the guarantees of \sketchmatrix\ and \recovernorm\ provided by \cref{thm:SketchApprox} to prove the following lemma.

\ApxLoad*
\begin{proof}
Note that $\GetApproxLoad$ sends $(L\pinv B\tp C\tp)\tp b_e$ to $\recovernorm$. This simplifies to $CBL\pinv b_e$. We get from the approximation guarantee of \cref{thm:SketchApprox} that
\[
\ome \cdot \| BL\pinv b_e \|_1 \le \recovernorm(BL\pinv b_e) \le \ope \cdot \| BL\pinv b_e \|_1
\]
Since $\loadw(e) = w_e \cdot \|BL\pinv b_e\|_1$, multiplying the above inequality by $w_e$, we get
\[
\ome \cdot \loadw(e) \le \GetApproxLoad(G, w) \le \ope \cdot \loadw(e)
\]
as required.
\end{proof}

\section{Running time analysis}
\label{sec:runtime}

In this section, we show that Algorithm~\ref{alg:MWURoute} (\ComputeRouting) runs in time $\tO(m^{3/2})$. If we show that each iteration of the for loop in \ComputeRouting\ takes time $\tO(m)$, then since $T= O(\sqrt{m})$, the claimed running time then follows.
We first show that Algorithm~\ref{alg:xLoadApx} (\GetApproxLoad) runs in time $\tO(m)$.

\begin{restatable}{lemma}{GetApxLoadTime}
\label{lem:GetApxLoadTime}
Algorithm~\ref{alg:xLoadApx} runs in time $\tO(m)$.
\end{restatable}
\begin{proof}
Calculating $B$ and $L$ takes time $O(m)$. By \cref{thm:SketchApprox}, building the sketch matrix $C$ takes time $O(\ell m)$.
Since we set $\sketchd = n^{-10}$ and $\sketche = 1/2$ for the sketch matrix, we get $\ell = O(\log n)$, which gives $O(\ell m) = \tO(m)$.

Each row of $B\tp C\tp$ can be obtained as follows: Since the row of $B\tp$ corresponding to vertex $u$ has $\deg(u)$ non-zero entries, the row of $B\tp C\tp$ corresponding to vertex $u$ can be obtained by taking the sum of every row of $C\tp$ corresponding to an edge that is incident to $u$.
Since each row of $C\tp$ has $\ell$ entries, this involves $O(\deg(u) \cdot \ell)$ calculations for computing row $u$ of $B\tp C\tp$. Since $\sum_u \deg(u) = 2m$, this gives an overall bound of $\tO(m)$ for calculating $B\tp C\tp$.
Each Laplacian solver takes time $\tO(m)$, and since we solve for $\ell$ vectors, all Laplacian solves together take time $\tO(m)$ as well.

Finally, we perform \recovernorm\ for $m$ edges. Each invocation of \recovernorm\ takes time $O(\ell)$ by \cref{thm:SketchApprox}, and thus all calls to \recovernorm\ together run in time $\tO(m)$. This proves the lemma.
\end{proof}

We can use this to show that each iteration of the for loop in $\ComputeRouting$ runs in time $\tO(m)$, and thus the entire algorithm runs in time $\tO(m^{3/2})$.

\begin{restatable}{lemma}{ComputeRoutingTime}
\label{lem:ComputeRoutingTime}
Algorithm~\ref{alg:MWURoute} runs in time $\tO(m^{3/2})$.
\end{restatable}
\begin{proof}
For each iteration of the for loop, normalizing the $\xe{t-1}$s and calculating $\we{t}$ (and thus $\WW{t}$) takes time $O(m)$. Calculating approximate loads is $\tO(m)$ by \cref{lem:GetApxLoadTime}. Since computing $\xe{t}$ takes time $O(m)$, each iteration runs in time $\tO(m)$. Thus the algorithm runs in time $\tO(Tm) = \tO(m^{3/2})$ by choice of $T = O(\sqrt{m} )$.
\end{proof}

\section{\texorpdfstring{Extension to the $\ell_1$ norm Oblivious Routing}{Extension to the l\_1 norm Oblivious Routing}}
\label{sec:l1}

In this section we prove that a convex combination of $O(\sqrt{m})$ electrical routings gives an oblivious routing scheme with respect to the $\ell_1$ norm that achieves a $O(\log^2 n)$ competitive ratio.
By the characterization of the competitive ratio as $\beta_1 = \|MB\tp\|_1$, our goal is now to bound the maximum \emph{stretch} of an edge.

Consider the oblivious routing operator $M_w = WBL\pinv$. The \emph{stretch} of an edge $e \in E$ with respect to the routing operator $M_w$ is given by:
\[
    \stre_w(e) := \sum_{f} w_f |b_e L\pinv b_f\tp |.
\]

Similar to earlier, our goal is to solve the following linear program.

\begin{equation*}\label{primal}%
\begin{array}{@{}ll@{}rl}
 \text{min}  &  \alpha &  &   \\[0.2cm]
 \text{s.t.}  & \forall e \quad  \quad &   \sum_{i} \nolimits \lambda_i \cdot \stre_{w_i}(e) &\leq \alpha \\[0.3cm]
  & &  \sum_{i} \nolimits \lambda_i & = 1  \\[0.3cm]
  & \forall i \quad \quad &  \lambda_i & \geq 0.
\end{array}
\end{equation*}

Doing the same dual construction as for the load gives us that the dual is equivalent to solving
\begin{equation}
\label{eq:strebetastar}
\max_{p \in \Delta^m} \min_i \sum_e p_e \stre_{w_i}(e),
\end{equation}

Then the only difference from earlier is having to show that we can bound the average stretch and the width for stretch as well.
In particular, we first show that for every $p\in \Delta^m$ we can produce a weighting $w$ such that electrical routing on $G$ with weights $w$ gives low average stretch, i.e., $\sum_e p_e \stre_w \le \alocal$.

\begin{restatable}{lemma}{StretchAverage}
\label{lem:strechOnAverage}
For any probability distribution $p \in \Delta^m$, the oblivious routing $M_w$ corresponding to the electrical network with weights $w_e = (p_e +  \nicefrac{1}{m})$ satisfies $\sum_{e} p_e \stre_w(e) \leq 2\alocal$.
\end{restatable}
\begin{proof}
Setting $\ell_e$ to $\sqrt{w_e}$ and applying~\cref{lem:Localization}, we get
\begin{align*}
\sum_e p_e \stre_w(e)
&= \sum_e p_e \cdot \sum_{f} w_f \cdot \stre_w(e\to f) \\
&= \sum_e p_e \cdot \sum_{f}  w_f \cdot |b_e L\pinv b_f\tp| \tag*{(by definition of stretch)}\\
&\leq \sum_e (p_e + \nicefrac{1}{m}) \cdot \sum_f w_f \cdot |b_e L\pinv b_f\tp| \\
&\leq \sum_e w_e \cdot \sum_{f} w_f \cdot |b_e L\pinv b_f\tp| \tag*{(by definition of $w_e$)}\\
&= \sum_e \sum_{f} w_e \cdot w_f \cdot |b_e L\pinv b_f\tp| \\
&= \sum_{e, f} \sqrt{w_e w_f} \cdot \sqrt{w_e w_f} \cdot |b_e L\pinv b_f\tp | \\
&= \sum_{e, f} \ell_e \ell_f \cdot \sqrt{w_e w_f} \cdot |b_e L\pinv b_f\tp | \tag*{(by choice of $\ell_e$)}\\
&\le \alocal \cdot \|\ell\|_2^2 \tag*{(by \Cref{lem:Localization})} \\
&= \alocal \cdot \sum_e {w_e} \\
&= \alocal \cdot \sum_e \left( p_e + \nicefrac{1}{m} \right) \\
&= 2 \alocal,
\end{align*} as required.
\end{proof}

We next bound the width by
showing an upper bound on the stretch of every single edge.

\begin{restatable}{lemma}{WidthBoundL1}
\label{lem:WidthBoundL1}
For any probability distribution $p \in \Delta^m$, the oblivious routing $M_w$ corresponding to the electrical network with weights $w_e = (p_e + \nicefrac{1}{m})$ satisfies $\stre_w(e) \le \sqrt{2m}$ for every edge $e$.
\end{restatable}
\begin{proof}
Fixing an edge $e$,
\begin{align*}
\stre_w(e)
&=  \sum_f w_f |b_e L\pinv b_f\tp| \\
&\le \sum_f \sqrt{\nicefrac{w_f}{w_e}} \cdot \sqrt{w_e w_f} |b_e L\pinv b_f\tp| \\
&\le \sqrt{\sum_f \nicefrac{w_f}{w_e}} \cdot \sqrt{\sum_f w_e w_f |b_e L\pinv b_f\tp|^2} \tag*{(by Cauchy-Schwarz)}
\end{align*} We now bound each of the two terms separately. For the first term,
\begin{align*}
\nicefrac{1}{w_e} \cdot \sum_f w_f
&= \weight \cdot \sum_f \left( p_f + \nicefrac{1}{m}\right) \\
&\le 2 \cdot \weight \\
&\le 2 \cdot \nicefrac{1}{\left(\nicefrac{1}{m}\right)} \\
&= 2m.
\end{align*}

For the second term, by Lemma~\ref{lem:widthHelpfulLemma}, we know that
\[
    \sum_f w_e w_f |b_e L\pinv b_f \tp|^2 \leq 1.
\]

Putting these two inequalities together, we get that $\stre_w(e) \le \sqrt{2m}$ for any edge $e$, which gives the desired bound of $\sqrt{2m}$ on the width.
\end{proof}

With these two lemmas, the rest of the proofs of the bound on the competitive ratio are exactly the same as in the $\ell_{\infty}$ case.
While most of the running time analysis holds, note that we need to sketch slightly different matrices now. Earlier, we wanted to approximate $\loadw(e) = w_e \cdot \sum_f \aeLf$. Thus we approximated $\|BL\pinv b_e\|_1$ with our sketch matrix, and then multiplied it with $w_e$ to obtain approximate loads. Since we now need to approximate $\strew(e) = \sum_f w_f \aeLf$, we instead sketch $WBL\pinv b_e$ to obtain the approximate stretch.
Note that $\load$ and $\stre$ are simply the row and column 1-norms respectively of $WBL\pinv B\tp$.

\section{The Representation of Oblivious Routing and the Parallel Complexity}
\label{sec:repPRAM}

In this section, we discuss the representation of our oblivious routing scheme and the parallel complexity of our algorithms.

\paragraph*{Representation of Oblivious Routing.} Any \emph{linear} oblivious routing scheme can be implemented using the following representation: every edge $e \in E$  stores the flow sent across edge $e$ by an oblivious routing that sends one unit of flow
from $u$ to $x$, for every vertex $u \in V$ and some arbitrary but fixed target vertex $x \in V$. We let $f_{u,x}(e)$ denote the value of such an oblivious flow. Thus every edge stores $n$ values and the total space requirement is $O(nm)$.  Upon receiving a query for routing demand pairs ${\dem_{s,t}}$ of some demand vector $\dem \in \mathbb{R}^{\binom{n}{2}}$, we can compute the $(s,t)$-flow along any edge $e \in E$ by computing
\[
      \dem_{s,t} \cdot (f_{s,x}(e) - f_{t,x}(e)),
\]
where the correctness follows from the fact that the oblivious routing operator is linear. %

We need to show that our oblivious routing operator based on a convex combination of electrical routings is linear. Note that $M_w := WBL \pinv$ is a linear routing operator for any weighting on the edges $\{w_e\}$. Therefore, any convex combination $\sum_{i} \lambda_i M_{w_i}$ with $\sum_{i} \lambda_i = 1$ is also a linear routing operator.

It remains to study the running time of constructing such a representation, which we refer to as the \emph{preprocessing time}. We start by considering the cost of constructing the representation for a single electrical routing $M_w := WBL \pinv$. Since $f_{u,x}(e) = (WBL\pinv \chi_{u,x})_e$, our goal is to compute $WBL\pinv \chi_{u,x}$, for every $u \in V$ and the fixed vertex $x$, which can be achieved by solving $n$ Laplacian systems. By Theorem~\ref{thm:LapSolve}, each such system can be solved in $\tilde{O}(m)$ time, which in turn leads to a processing time of $\tilde{O}(m n)$ for a single electrical routing. Since our oblivious routing scheme consists of $O(\sqrt{m})$ electrical routings, it follows that the total preprocessing time for constructing the representation for these electrical routings is $\tilde{O}(m^{3/2}n)$.

\paragraph*{Parallel Complexity.}

We next bound the parallel complexity of (i) the multiplicative weights updates algorithm for computing the weights of our oblivious routing scheme (Algorithm~\ref{alg:MWURoute}) and (ii) the algorithm for computing the representations of our scheme.
Both results rely on the fact that a Laplacian system can be solved to high accuracy with nearly-linear work and polylogarithmic depth.

\begin{theorem}[\cite{PengS14,KyngLPSS16}]
\label{thm:parallelLaplacianSolver}
Given a $n$-vertex $m$-edge graph $G$, a Laplacian matrix $L$, a demand vector $y \in \mathbb{R}^{n}$ and an error bound $\epsilon_L$, there is a parallel algorithm that achieves $\tilde{O}(m)$ work and $\tilde{O}(1)$ depth and returns a vector $x \in \mathbb{R}^{n}$ such that
\[
    \|x - L\pinv y\|_L \le \lsolvee \cdot \|L\pinv y\|_L.
\]
\end{theorem}

Our first result shows that our MWU-based oblivious routing can be implemented in parallel with $\tilde{O}(m^{3/2})$ work and $\tilde{O}(\sqrt{m})$ depth.

\begin{lemma}
\label{lem:parallelMWU}
There is a parallel implementation of Algorithm~\ref{alg:MWURoute} that achieves $\tilde{O}(m^{3/2})$ work and $\tilde{O}(\sqrt{m})$ depth.
\end{lemma}
\begin{proof}
We start by analyzing the parallel complexity of Algorithm~\ref{alg:MWURoute}. Consider one iteration of the \textbf{for} loop, and observe that the parallel complexity is dominated by the  parallel cost of computing approximate loads in Line 9, i.e., Algorithm~\ref{alg:xLoadApx}. Hence, it suffices to bound the parallel complexity of the latter. For generating the sketch matrix $C$, note that each edge $e \in E$ draws $\ell = O(\log n)$ independent random variables from the Cauchy distribution~\cite{Ind06sketch}, and since these operations can be performed locally, it follows that constructing $C$ takes $\tilde{O}(m \ell) = \tilde{O}(m)$ work and $O(1)$ depth.

Next, we compute the $(n \times \ell)$-dimensional matrix $X = B \tp C\tp$ in a row-wise fashion (Algorithm~\ref{alg:xLoadApx}, Line 4). The $u$-th row of $B \tp$ contains $\deg(u)$ non-zero entries, and thus an entry $(X)_{u,i} = (B \tp C\tp)_{u,i} = \sum_{e \mid e \sim u} b_{u,e}\tp c_{e,i}\tp$ can be evaluated with $O(\deg(u))$ work and $O(\log n)$ depth. As $X$ has only $\ell = O(\log n)$ columns, it follows that $u$-th row can be computed with $\tilde{O}(\deg(u))$ work and $O(\log n)$ depth. Summing the costs over all rows of $X$, we conclude that $X$ can be computed with $\tilde{O}\left(\sum_{u} \deg(u)\right) = \tilde{O}(m)$ work and $O(\log n)$ depth. Finally, Lines 5 and 6 of Algorithm~\ref{alg:xLoadApx} involve solving $\ell=O(\log n)$ Laplacian systems. By Theorem~\ref{thm:parallelLaplacianSolver}, we can solve all these systems in parallel with $\tilde{O}(m)$ work and $\tilde{O}(1)$ depth.

Bringing all the above bounds together shows that an iteration of the \textbf{for} loop in Algorithm~\ref{alg:MWURoute} can be implemented in parallel with $\tilde{O}(m)$ work and $\tilde{O}(1)$ depth. Since in total there are $\tilde{O}(\sqrt{m})$ iterations, we get that a parallel implementation of Algorithm~\ref{alg:MWURoute} has $\tilde{O}(m^{3/2})$ work and $\tilde{O}(\sqrt{m})$ depth.
\end{proof}

Our second result show that the representation of our oblivious routing can be implemented in parallel with $\tilde{O}(m^{3/2}n)$ work and $\tilde{O}(1)$ depth.

\begin{lemma}
The representation of the oblivious routing based on $O(\sqrt{m})$ electrical routings can be implemented in parallel with $\tilde{O}(m^{3/2}n)$ work and $\tilde{O}(1)$ depth.
\end{lemma}
\begin{proof}
We first analyze the cost of a single electrical routing given a weighting of the edges. Recall that our representation requires that every vertex solves one Laplacian system. These systems can be solved independently of each other (and thus in parallel), and by Theorem~\ref{thm:parallelLaplacianSolver}, each of them can be implemented in parallel with $\tilde{O}(m)$ work and $\tilde{O}(1)$ depth. Thus, the parallel complexity of a single electrical routing is $\tilde{O}(mn)$ work and $\tilde{O}(1)$ depth.

Now, note that our MWU algorithm has already computed $\tilde{O}(\sqrt{m})$ weightings of the graph, each corresponding to a single electrical routing. Once computed, these weightings are independent of each other and, thus, the electrical routings (one per vertex) for all of these weightings can be computed in parallel. Therefore, the total complexity for computing the representation of our routing scheme is $\tilde{O}(m^{3/2}n)$ work and $\tilde{O}(1)$ depth.

To evaluate the average (oblivious) flow from the convex combination of $\tilde{O}(\sqrt{m})$ electrical routings, each edge can locally compute the convex combination of the oblivious flows sent along that edge  with $\tilde{O}(\sqrt{m})$ work and $O(\log(\sqrt{m})) = \tilde{O}(1)$ depth. Thus, it follows that the total parallel complexity of this step is $\tilde{O}(m^{3/2})$ work and $\tilde{O}(1)$ depth.

Bringing the above bounds together proves the lemma.
\end{proof}

\paragraph*{Funding}

\begin{wrapfigure}{r}{0.15\textwidth}
\includegraphics[width=0.13\textwidth]{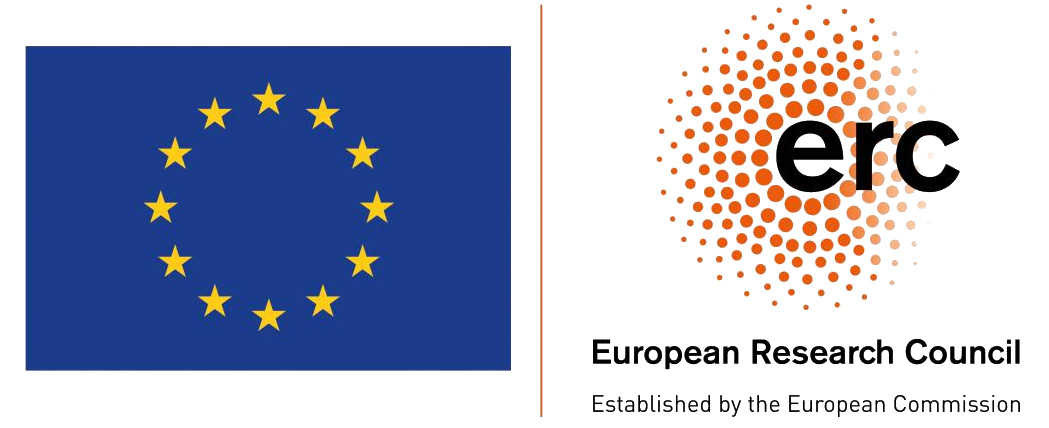}
\end{wrapfigure}

\textit{M. Henzinger and A. R. Sricharan}:
This project has received funding from the European Research Council (ERC) under the European Union's Horizon 2020 research and innovation programme (Grant agreement No. 101019564) and the Austrian Science Fund (FWF) project Z 422-N, project
I 5982-N, and project P 33775-N, with additional funding
 from the \textit{netidee SCIENCE
Stiftung}, 2020--2024.

\textit{S. Sachdeva}:
SS's work is supported by an Natural Sciences and Engineering Research Council of Canada (NSERC) Discovery Grant RGPIN-2018-06398 and a Sloan Research Fellowship.

\textit{H. R\"acke}:
Research supported by German Research Foundation (DFG), grant 470029389 (FlexNets), 2021-2024.
\bibliographystyle{beta}
\def\bibfont{\small}
\bibliography{Ref, references}

\appendix

\section{Algorithm using approximate Laplacian solvers}
\label{sec:apxlsolve}

\newcommand{\allones}{\ensuremath{\mathds{1}}}
\newcommand{\med}{\ensuremath{m}}
\newcommand{\epschoice}{\ensuremath{\nicefrac{\epsilon}{( 8m n^4 K)}}}
\newcommand{\adde}{\ensuremath{\nicefrac{\epsilon}{2 mn}}}
\newcommand{\addeload}{\ensuremath{\nicefrac{2}{m}}}
\newcommand{\ubound}{\ensuremath{4\loade n^3 K}}

We give in Algorithm~\ref{alg:LoadApx} a version of $\GetApproxLoad$ that uses an approximate Laplacian solver~\cite{ST04} instead of an exact one. The change is in Line~\ref{line:apxlapsolve}. For concreteness, we use the following approximate Laplacian solver by Jambulapati and Sidford~\cite{JS21laplacian}.

\begin{algorithm}[t]
\SetAlgoLined
\DontPrintSemicolon
\caption{\GetApproxLoad, to compute approximate loads for electrical routing.}
\label{alg:LoadApx}
\KwInput{A graph $G$, weights $\{ w_e \}_{e \in E}$ on the edges, approximation factor $\loade$.}
\KwOutput{Approximation $\{ \apxloadw \}_{e \in E}$ to the load on the edges.}
Let $B$ be the incidence edge-vertex matrix of $G$\;
Let $L := B\tp \diag(w) B$ be the Laplacian matrix\;
Set $C \leftarrow \sketchmatrix(m, n^{-10}, \nicefrac{\loade}{2})$\;
Set $X \leftarrow B\tp C\tp$\;
Let $X^{(i)}$ be the $i^{th}$ column of $X$ for all $i \in [\ell]$\;
Set $\tU^{(i)} \leftarrow \lapsolve(L, X^{(i)}, n^{-10}, \lsolvee)$ for all $i \in [\ell]$\label{line:apxlapsolve}\;
Set $\tU \leftarrow (\tU^{(1)}, \tU^{(2)}, \ldots, \tU^{(\ell)})$ \algcomment{$\tU \approx (CBL\pinv)\tp$}
Set $\apxloadw(e) \leftarrow w_e \cdot \recovernorm(\tU\tp b_e)$ for all $e \in E$\;
\Return{$\apxloadw$}
\end{algorithm}

\begin{restatable}[\cite{JS21laplacian}]{theorem}{LapSolve}
\label{thm:LapSolve}
There is an algorithm $\lapsolve$ that gets as input a Laplacian $L$ of an $n$-vertex $m$-edge graph, a vector $y \in \R^n$, error bound $\lsolvee$, and returns a vector $x \in \R^n$ such that the following holds with probability $\ge 1 - \nicefrac{1}{\poly(n)}$.
\[
\|x - L\pinv y\|_L \le \lsolvee \cdot \|L\pinv y\|_L
\] where $\|x\|_L = \sqrt{x\tp L x} $ is the norm induced by the Laplacian. Further, the algorithm runs in time $\tO(m \log(1/\lsolvee))$.
\end{restatable}

\begin{restatable}{lemma}{ApxSolveLoad}
\label{lem:ApxSolveLoad}
For any approximation factor $0 < \loade < 1$, and any weighted graph $(G,w)$, let $\loadw = \left( w_e \sum_f \aeLf \right)_{e \in E} $ be the true loads, and $\apxloadw = \GetApproxLoad(G, w, \loade)$ be the approximate loads computed by the algorithm when using an approximate Laplacian solver. Then with probability $\ge 1-\nicefrac{1}{\poly(n)}$,
\[
(1 - \loade) \cdot \loadw(e)  \le \apxloadw(e) \le (1 + \loade) \cdot \loadw(e) \qquad \text{for all $e \in E$}
\]
\end{restatable}

We prove this lemma in \cref{sec:proofGetApxSolveLoad}, and then explain the changes to the algorithm in \cref{sec:changes}.

\subsection{\texorpdfstring{Proof of \cref{lem:ApxSolveLoad}}{Proof of Lemma ApxSolveLoad}}
\label{sec:proofGetApxSolveLoad}

We assume without loss of generality that $L\pinv x \perp \allones$ and $y \perp \allones$, where $\allones$ is the the all-ones vector, $x$ is the input to $\lapsolve$ and $y$ is the output received. This is because shifting a vector of potentials by a constant vector does not affect the guarantees in the $L$ norm.

We will use the following two properties of $\sketchmatrix$ and $\recovernorm$ in our analysis:
\begin{itemize}
    \item $\recovernorm(x_1, x_2, \ldots, x_{\ell}) = \medianx$.
    \item $\max_{i, j} C_{ij} \le \poly(m)$.
\end{itemize}

We first prove \cref{lem:ApxSolveLoad} assuming the following lemmas.
The first lemma converts the guarantee we have on the $L$ norm to a guarantee on the $\ell_{\infty}$ norm.
\begin{restatable}{lemma}{InftyClose}
\label{lem:InftyClose}
Suppose we have two vectors $x, y \in \R^n$ such that $x, y \perp \allones$ and $\|x - y\|_L \le \epsilon \|y\|_L$. Then $\|x - y\|_{\infty} \le \epsilon \cdot \totalb \cdot \|y\|_{\infty}$.
\end{restatable}

Once we have this $\ell_{\infty}$ guarantee for the columns of our matrix $\tU$, we convert this to a guarantee on the rows of $\tU$ using the following lemma.

\begin{restatable}{lemma}{RowtoColumnInfty}
\label{lem:RowtoColumnInfty}
Let $U$ be an $n \times \ell$ matrix, with $\largest$ being the largest value present in the matrix. Let $\tU$ be an approximation of $U$ such that for every $j \in [\ell]$, the $j$-th column $\tC_j$ of $\tU$ is $\epsilon$-close to the $j$-th column $C_j$ of $U$ in the following sense: $\forall j \in [\ell],  \|\tC_j - C_j\|_{\infty} \le \epsilon \cdot \|C_j\|_{\infty}$. Then for any $i \in [n]$, and rows $R_i$ of $U$ and $\tR_i$ of $\tU$,
\[
\| \tR_i - R_i \|_{\infty} \le \epsilon \cdot \largest
\]
\end{restatable}

We use the following lemma to show that $\recovernorm$ works well on approximate sketched vectors.

\begin{restatable}{lemma}{CombinedApx}
\label{lem:CombinedApx}
Let $a \in \R$ be a real number. Suppose $x \in \R^\ell$ satisfies
\[
\ome \cdot a \le \medianx \le \ope \cdot a
\]
Suppose $y \in \R^\ell$ is such that $\|x - y\|_{\infty} \le \edash$, then $y$ satisfies
\[
\ome \cdot a - \edash \le \mediany \le \ope \cdot a + \edash
\]
\end{restatable}

We finally use a lower bound on the effective resistance to convert an additive approximation guarantee into a multiplicative one. This follows from Rayleigh's monotonicity law and is shown, for example, by Spielman and Srivastava~\cite[Proposition 10]{SS11sparse}.

\begin{restatable}{lemma}{LoadLB}
\label{lem:LoadLB}
For any weighted graph $(G, w)$, the load on any edge $e$ has the following lower bound.
\[
\loadw(e) \ge \frac{2w_{e}}{nw_{\max}},
\] where $w_{\max}$ is the maximum edge weight.
\end{restatable}

We first use these four lemmas to prove \cref{lem:ApxSolveLoad}.

\ApxSolveLoad*
\begin{proof}
Let $U = L\pinv B\tp C\tp$, and $\tU = \lapsolve(L, B\tp C\tp, \lsolved, \lsolvee)$ where we abuse notation to mean that $\lapsolve$ runs on each column of $B\tp C\tp$ and returns a column vector.
Let $R_i$ and $\tR_i$ denote the rows,  and $C_j$ and $\tC_j$ denote the columns of $U$ and $\tU$ respectively.
By guarantees of the approximate Laplacian solver, we have
\[
\|\tC_j - C_j\|_L \le \lsolvee \|C_j\|_L
\]
\cref{lem:InftyClose} tells us that
\[
\|\tC_j - C_j\|_{\infty} \le \lsolvee \cdot \totalb \cdot \|C_j\|_{\infty}
\]
With $K$ as an upper bound on the maximum entry in $CBL\pinv$, and using the above inequality in \cref{lem:RowtoColumnInfty}, we get that
\[
\|\tR_i - R_i\|_{\infty} \le \lsolvee \cdot \totalb \cdot \largest
\]
Note that for any $e = (u, v)$, $(U\tp b_e)\tp = R_u - R_v$, and $(\tU\tp b_e)\tp = \tR_u - \tR_v$. Thus
\begin{align*}
\|(U\tp b_e) - (\tU\tp b_e)\|_{\infty}
&= \|(R_u - R_v) - (\tR_u - \tR_v)\|_{\infty} \\
&\le \|R_u - \tR_u\|_{\infty} + \|R_v - \tR_v\|_{\infty} \\
&\le \ubound
\end{align*}
By choice of $\lsolvee = \epschoice$, the final expression is $\adde$.
Since $\recovernorm(x_1, x_2, \ldots, x_\ell) = \medianx$, and since the guarantees of $\recovernorm$ gives us
\[
\omet \cdot \sum_f \aeLf \le \recovernorm(U\tp b_e) \le \opet \cdot \sum_f \aeLf,
\] using \cref{lem:CombinedApx} on $U\tp b_e$ and $\tU\tp b_e$, we get
\[
\omet \cdot \sum_f \aeLf - \adde \le \recovernorm(\tU\tp b_e) \le \opet \cdot \sum_f \aeLf + \adde.
\]
Thus,
\[
\omet \cdot \load_w(e) - \nicefrac{\epsilon w_e}{2mn} \le \GetApproxLoad(G, w) \le \opet \cdot \load_w(e) + \nicefrac{\epsilon w_e}{2mn}
\]
Since for any edge $f \in E,$ we have $w_f = \left( p_f + \nicefrac{1}{m} \right)^{-1} \le m$,  we get $w_{\max} \le m.$ Thus, \cref{lem:LoadLB} gives us that $\loadw(e) \ge \nicefrac{w_e}{mn},$ which proves our lemma.
\end{proof}

We now prove the intermediate lemmas.
\begin{restatable}{lemma}{InftyLUB}
\label{lem:InftyLUB}
For any $x \in \R^n$, $\|x\|_L \le \ub \cdot \|x\|_{\infty}$.
\end{restatable}
\begin{proof}
Since $\|x\|_{\infty}^2 \le 1$ implies $|x_i|\le 1$ for all $i \in [n]$, and since $\|x\|_L^2 = x\tp L x$, we have
\begin{equation*}
\max_{x \in \R^n} \frac{\|x\|_L^2}{\|x\|_{\infty}^2}
\le \max_{x \in \R^n: |x_i| \le 1} x\tp L x
\end{equation*}
Since $x\tp L x = \sum_{(u,v) \in E} (x_u - x_v)^2$, and $|x_i| \le 1,$ each term is bounded by 4. Thus, the sum is at most $4m \le 4n^{2},$ giving the lemma.
\end{proof}

\begin{restatable}{lemma}{InftyLLB}
\label{lem:InftyLLB}
For any $x \in \R^n$ such that $x \perp \allones$, we have $\lb \cdot \|x\|_{\infty} \le \|x\|_L$.
\end{restatable}
\begin{proof}
Since $\min_{x \in \R^n} \|x\|_L^2 / \|x\|_2^2 = 0$ which is attained by $\allones$, by the variational characterization of eigenvalues of a symmetric matrix, the second smallest eigenvalue $\lambda_2$ of $L$ is given by
\[
\min_{x \in \R^n: x \perp \allones} \frac{\|x\|_L^2}{\|x\|_2^2} = \lambda_2
\]
By Cheeger's inequality, $\lambda_2 \ge h^2/2\Delta$, where $h = \min_{S : |S|  \le n/2} |E(S, \overline{S})|/|S|$ is the unnormalized Cheeger's constant and $\Delta$ is the maximum degree in the graph. Thus
\[
\lambda_2 \ge h^2 \cdot \frac{1}{2\Delta} \ge \frac{4}{n^2} \cdot \frac{1}{2n} \ge n^{-3}
\]
Thus for any $x \in \R^n$ such that $x \perp \allones$, $\|x\|_L \ge n^{-3/2} \cdot \|x\|_2 \ge n^{-3/2} \cdot  \|x\|_{\infty}$, and the lemma follows.
\end{proof}

\InftyClose*
\begin{proof}
By \cref{lem:InftyLUB} and \cref{lem:InftyLLB}, we have
\[
\lb \cdot \|x-y\|_{\infty} \le \|x-y\|_L
\qquad\text{and}\qquad \|y\|_L \le \ub \cdot \|y\|_{\infty}
\] Together with the assumption in our lemma, we get
\[
\|x - y\|_{\infty} \le \lbnm \cdot \|x - y\|_L \le \epsilon \cdot \lbnm \cdot \|y\|_L \le \epsilon \cdot  \totalb \cdot \|y\|_{\infty}
\] as claimed.
\end{proof}

\RowtoColumnInfty*
\begin{proof}
For any $i \in [n]$,
\begin{align*}
\|\tR_i - R_i\|_{\infty}
&\le \max_{i \in [n]} \|\tR_i - R_i\|_{\infty} \\
&= \max_{i, j \in [n] \times [\ell]} | \tU_{ij} - U_{ij} | \\
&= \max_{j\in [\ell]} \| \tC_{j} - C_{j} \|_{\infty} \\
&\le \epsilon \cdot \max_{j \in [\ell]} \| C_{j} \|_{\infty} \\
&\le \epsilon \cdot \largest
\end{align*}
as required.
\end{proof}

\begin{restatable}{lemma}{MedianApx}
\label{lem:MedianApx}
Let $x, y \in \R^\ell$ be such that $\|x - y\|_{\infty} \le \epsilon$, where $\ell$ is odd. Then
\[
\medianx - \epsilon \le \mediany \le
\medianx + \epsilon
\]
\end{restatable}
\begin{proof}
Note that if $|x_i - y_i| \le \epsilon$, then $| |x_i| - |y_i| | \le \epsilon$ as well. Thus we can assume without loss of generality that $x$ and $y$ are non-negative.

Relabel the indices such that $x_1 \le x_2 \le \ldots \le x_\ell$ are in non-decreasing order. The assumption on $\ell_{\infty}$ norm gives us that for all $i \in [\ell]$,
\begin{equation}
\label{eq:MedianApx}
|x_i - y_i| \le \epsilon
\end{equation}
Let $\pi$ be a permutation such that $y_{\pi(1)} \le y_{\pi(2)} \le \ldots \le y_{\pi(\ell)}$ is in non-decreasing order. Let $\med = \left\lceil \ell/2 \right\rceil $. We want to show that $|y_{\pi(\med)} - x_{\med}| \le \epsilon$.
If $\med = \pi(\med)$, then the lemma follows from \cref{eq:MedianApx} with $i = \med$.
So assume that $\med \neq \pi(\med)$.

Consider the case when $\med > \pi(\med)$. Then since $y_{\med} \ge y_{\pi(\med)}$, we have $x_{\med} \ge y_{\med} - \epsilon \ge y_{\pi(\med)} - \epsilon$. For the other direction, since $\med > \pi(\med)$, there exists an index $j$ such that $x_j > x_{\med}$ and $y_j \le y_{\pi(\med)}$. For this index, we have $y_{\pi(\med)} \ge y_j \ge x_j - \epsilon \ge x_m - \epsilon$ as required.
The case when $\med < \pi(\med)$ is symmetric.
\end{proof}

\CombinedApx*
\begin{proof}
Using \cref{lem:MedianApx} and the assumption on $x$,
\begin{align*}
\ome \cdot a - \edash
&\le \medianx - \edash \\
&\le \mediany \\
&\le \medianx + \edash \\
&\le \ope \cdot a + \edash
\end{align*}
as required.
\end{proof}

\subsection{\texorpdfstring{Changes in the algorithm}{Changes in the algorithm}}
\label{sec:changes}

The algorithm is very similar to the case with an exact Laplacian solver, with minor changes. We detail them for completeness.
Since the maximum entry in $C$ is upper bounded by $\poly(n, m)$, the maximum entry in $CBL\pinv$ is upper bounded by some $\poly(n,m)$ as well. We call this value as $K = \max_{i,j} (CBL\pinv)_{ij}$.
We use $\lsolvee = \epschoice$.

The only change in the running time from the earlier algorithm is the time taken by the approximate Laplacian solver. Recall that the running time of the approximate Laplacian solver is $\tO(m \log(1/\lsolvee))$.
Since the maximum entry in $C$ is bounded by $\poly(m)$, the maximum entry in $CBL\pinv$ is bounded by $\poly(n,m)$ as well. We denote this upper bound as $K$. Thus $\lsolvee \ge 1/\poly(n)$, and the approximate Laplacian solver runs in time $\tO(m)$.

\section{Routing on capacitated graphs}
\label{sec:capacitated}

We detail the changes to obtain an oblivious routing on capacitated graphs.

\subparagraph*{Capacitated Graph}
A capacitated graph $G = (V, E, u)$ is a undirected graph along with a function $u : E \to \R^+$ that represents the capacity of each edge.

\subparagraph*{Congestion}
Given a flow $f \in \R^m$, the congestion of an edge is the amount of flow on that edge relative to its capacity, given by ${|f_e|}/{u_e}$.

Let $U$ denote the $m \times m$ diagonal matrix with $u_e$ on the diagonals.
Kelner and Maymounkov~\cite[Theorem 3.1]{KM11electric} show that the worst-case demands for a capacitated graph is $u_e$ along each edge, i.e., the columns of $B\tp U$. Note that these can be routed optimally with congestion $1$, by simply routing each demand of $u_e$ across the same edge. In their presentation of the proof, they use $\{ w_e \}$ for both the capacities and the conductances. We, on the other hand, need to use conductances that are different from the capacities. While their proof works for our case, for the sake of clarity, we present their proof of worst-case demands in \cref{sec:worstcasedem} using our notation.
This then leads to the following definition of load.

\subparagraph*{Load}
For a linear oblivious routing $M$, the congestion of edge $e$ for routing $u_f \cdot b_f\tp$ is given by $ \load_M(f \to e) = u_e^{-1} \cdot u_f \cdot |(Mb_f\tp)_e|$. The load on edge $e$ is then given by summing this congestion up for each $f \in E$, giving
\[
\load_M(e) = \sum_f \load_M(f \to e) = u_e^{-1} \cdot \sum_f u_f \cdot |(M b_f\tp)_e|
\]
While $\textrm{cong}_M(e)$ would be better than $\load_M(e)$ in the capacitated case, we continue using $\load_M$ for continuity with the main body of the paper.
For an electrical flow with weights (i.e. conductances) $\left\{ w_e \right\} $, the corresponding oblivious routing is then $WBL\pinv$ (where the Laplacian is with respect to $W$, given by $B\tp W B$), and the load is
\[
\loadw(e) = \frac{w_e}{u_e} \cdot \sum_f u_f \cdot \aeLf
\]

To show the bound on competitive ratio for capacitated graphs, we then need to give a set of weights $ \left\{ w_e \right\} $ for each $p_e \in \Delta^m$ such that the MWU algorithm can be performed, and we need to show that we can still use sketching to get approximate loads in $\tO(m)$ time.

Concretely, we will use the weights $w_e = u_e^2 \cdot \weight$ in the algorithm. We need to prove analogues of \cref{lem:OracleReturn}, \cref{lem:WidthBound}, and we need to provide a version of \GetApproxLoad\ that works in the capacitated case. We do so in the next three sections. This, then, will give us \cref{thm:MainThm}.

\subsection{Bound on Average Loads}

\begin{restatable}{lemma}{CapOracleReturn}
\label{lem:CapOracleReturn}
For any probability distribution $p \in \Delta^m$, the oblivious routing $M_w$ corresponding to the electrical network with weights $w_e = u_e^2 \cdot \weight$ satisfies $\sum_e p_e \loadw(e) \le 2\alocal$.
\end{restatable}

\begin{proof}
Setting $\ell_e$ to $\nicefrac{u_e}{\sqrt{w_e}}$ and applying~\cref{lem:Localization}, we get
\begin{align*}
\sum_e p_e \loadw(e)
&= \sum_e p_e \cdot \sum_{f} \loadw(f\to e) \\
&= \sum_e p_e \cdot \nicefrac{w_e}{u_e} \cdot \sum_{f} u_f \cdot |b_e L\pinv b_f\tp| \tag*{(by definition of load)}\\
&= \sum_e p_e \cdot u_e \cdot \weight \cdot \sum_{f} u_f \cdot |b_e L\pinv b_f\tp| \tag*{(by definition of $w_e$)}\\
&\le \sum_e \sum_{f} u_e u_f \cdot |b_e L\pinv b_f\tp| \\
&= \sum_{e, f} \nicefrac{u_eu_f}{\sqrt{w_e w_f}} \cdot \sqrt{w_ew_f} \cdot |b_e L\pinv b_f\tp | \\
&= \sum_{e, f} \ell_e \ell_f \cdot \sqrt{w_e w_f} \cdot |b_e L\pinv b_f\tp | \tag*{(by choice of $\ell_e$)}\\
&\le \alocal \cdot \|\ell\|_2^2 \tag*{(by \Cref{lem:Localization})} \\
&= \alocal \cdot \sum_e \nicefrac{u_e^2}{w_e} \\
&= \alocal \cdot \sum_e \left( p_e + \nicefrac{1}{m} \right) \\
&= 2 \alocal,
\end{align*} as required.
\end{proof}

\subsection{Bound on Width}

\newcommand{\projC}{\ensuremath{U^{-\nicefrac{1}{2}}W^{\nicefrac{1}{2}} B L\pinv B\tp W^{\nicefrac{1}{2}}U^{\nicefrac{1}{2}}}}

We will use \cref{lem:widthHelpfulLemma} again in our proof of the following lemma.

\begin{restatable}{lemma}{CapWidthBound}
\label{lem:CapWidthBound}
For any probability distribution $p \in \Delta^m$, the oblivious routing $M_w$ corresponding to the electrical network with weights $w_e = u_e^2 \cdot \weight$ satisfies $\loadw(e)  \le \sqrt{2m}$ for every edge $e$.
\end{restatable}

\begin{proof}%
Note that by definition of $w_e$, we have $u_e = \sqrt{w_e \cdot \weightinv} $. Fixing an edge $e$,
\begin{align*}
\loadw(e)
&= \nicefrac{w_e}{u_e} \cdot \sum_f u_f \cdot |b_e L\pinv b_f\tp| \\
&\le \sqrt{\nicefrac{w_e}{\left(p_e + \nicefrac{1}{m}\right)}} \cdot \sum_f \sqrt{w_f \cdot \left( p_f + \nicefrac{1}{m} \right) }  \cdot |b_e L\pinv b_f\tp| \tag*{(by definition of $u_e$ and $u_f$)} \\
&\le \sum_f \sqrt{\nicefrac{\left( p_f + \nicefrac{1}{m} \right) }{\left(p_e + \nicefrac{1}{m}\right)}} \cdot \sqrt{w_e w_f}  \cdot |b_e L\pinv b_f\tp| \\
&\le \sqrt{\sum_f \nicefrac{\left( p_f + \nicefrac{1}{m} \right) }{\left(p_e + \nicefrac{1}{m}\right)}} \cdot \sqrt{\sum_f w_e w_f \cdot |b_e L\pinv b_f\tp|^2}  \tag*{(by Cauchy-Schwarz)}
\end{align*}

We now bound each of the two terms separately. For the first term, note that
\begin{align*}
\weight \cdot \sum_f \left( p_f + \nicefrac{1}{m}\right)
\le 2 \cdot \weight
\le 2 \cdot \nicefrac{1}{\left(\nicefrac{1}{m}\right)} = 2m.
\end{align*}

For the second term, by Lemma~\ref{lem:widthHelpfulLemma}, we know that
\[
    \sum_f w_e w_f |b_e L\pinv b_f \tp|^2 \leq 1.
\]

Putting these two inequalities together, we get that $\loadw(e) \le \sqrt{2m}$ for any edge $e$, which gives the desired bound of $\sqrt{2m}$ on the width.
\end{proof}

\subsection{Capacitated \GetApproxLoad}

\begin{algorithm}[th]
\SetAlgoLined
\DontPrintSemicolon
\caption{\GetApproxLoad, to compute approximate loads for electrical routing.}
\label{alg:capLoadApx}
\KwInput{A graph $G$, weights $\{ w_e \}_{e \in E}$ and capacities $\{ u_e \}_{e \in E}$ on the edges, approximation factor $\loade$.}
\KwOutput{Approximation $\{ \apxloadw(e) \}_{e \in E}$ to the load on the edges.}
Let $B$ be the edge-vertex incidence matrix of $G$\;
Let $L := B\tp \diag(w) B$ be the Laplacian matrix\;
Set $C \leftarrow \sketchmatrix(m, n^{-10}, \loade)$\;
Set $X \leftarrow B\tp U C\tp$\;
Let $X^{(i)}$ be the $i^{th}$ column of $X$ for all $i \in [\ell]$\;
Set $V^{(i)} \leftarrow \xlapsolve(L, X^{(i)})$ for all $i \in [\ell]$ \;
Set $V \leftarrow (V^{(1)}, V^{(2)}, \ldots, V^{(\ell)})$ \algcomment{$V = (CUBL\pinv)\tp$}
Set $\apxloadw(e) \leftarrow w_e u_e^{-1} \cdot \recovernorm(U\tp b_e)$ for all $e \in E$\;
\Return{$\apxloadw$}
\end{algorithm}

Algorithm~\ref{alg:capLoadApx} contains the changes to \GetApproxLoad\ for the capacitated case. Correctness follows from noting that $\loadw(e)$ is the $\ell_1$ norm of the $e^{th}$ row of $U^{-1}WBL\pinv B\tp U$.

\subsection{Worst-case demands}
\label{sec:worstcasedem}

We present the proof of the worst-case demands for oblivious routing on a capacitated graph $G=(V, E, u)$ being $B\tp U$ from Kelner and Maymounkov~\cite[Theorem 3.1]{KM11electric}, using $\{ u_e \}$ for the edge capacities, and $M$ as any linear oblivious routing.

Since congestion of a linear oblivious routing scales linearly when the demands are multiplicatively increased, it suffices to consider demands that can be routed optimally with congestion 1. Let $\{ \chi_i \}_i$ be some set of demands that can be optimally routed with congestion 1, and let $\{ f_i \}_i$ be such an optimal routing.
The claim follows from noting that demands $\{ \sum_i |f_{i,e}| \}_{e}$, i.e., demands of $\sum_i |f_{i,e}|$ across each edge $e$, can still be (non-linearly) routed with congestion 1.
\begin{align*}
\beta_{\infty}(M)
&\le \max\limits_{e'} u_{e'}^{-1} \cdot \sum\limits_{i} |b_{e'} M \chi_i| \tag*{for any $\{ \chi_i \}_i$ with $\OPT(\{ \chi_i \}_i)$ = 1} \\
&= \max\limits_{e'} u_{e'}^{-1} \cdot \sum\limits_{i} \left|b_{e'} M \left( \sum\limits_{e}f_{i,e} b_{e}\tp \right) \right| \tag*{(since $f$ routes $\{ \chi_i \}_i$)} \\
&= \max\limits_{e'} u_{e'}^{-1} \cdot \sum\limits_{i} \left|\sum\limits_{e}f_{i,e} \cdot b_{e'} M b_{e}\tp \right| \tag*{(by linearity of $M$)} \\
&\le \max\limits_{e'} u_{e'}^{-1} \cdot \sum\limits_{i,e} \left|f_{i,e} \cdot b_{e'} M b_{e}\tp \right| \tag*{(since $| \sum \cdot | \le \sum | \cdot |$)} \\
&= \max\limits_{e'} u_{e'}^{-1} \cdot \sum\limits_{e} \left| \sum\limits_i \left|f_{i, e}\right| \cdot b_{e'} M b_{e}\tp \right| \\
&\le \max\limits_{e'} u_{e'}^{-1} \cdot \sum\limits_{e} \left| u_e \cdot b_{e'} M b_{e}\tp \right| \tag*{(since $\{ f_i \}_i$ has congestion 1)} \\
&= \max\limits_{e'} u_{e'}^{-1} \cdot \sum\limits_{e} \left| \cdot b_{e'} M (u_e b_{e}\tp) \right| \\
&= \| U^{-1} M B\tp U \|_{\infty}
\end{align*}
as required.

\end{document}